\newif\ifSUBMIT
\SUBMITtrue
\SUBMITfalse

\newif\ifDRAFT
\DRAFTtrue
\DRAFTfalse

\ifSUBMIT
\documentclass[a4paper,review,fleqn]{cas-sc}
\usepackage{setspace}
\else
\documentclass[a4paper,final,fleqn]{cas-sc}
\ExplSyntaxOn
\cs_gset:Npn \__first_footerline:
  { \group_begin: \small \sffamily \__short_authors: \group_end: }
\ExplSyntaxOff
\fi

\usepackage[authoryear,longnamesfirst]{natbib}

\def\tsc#1{\csdef{#1}{\textsc{\lowercase{#1}}\xspace}}
\tsc{WGM}
\tsc{QE}
\usepackage{mathtools}
\usepackage{xspace}
\usepackage{etoolbox}

\usepackage{amsthm}      % provides theorem + proof
\usepackage{aliascnt}
\usepackage{cleveref}

\newtheorem{theorem}{Theorem}

\newaliascnt{lemma}{theorem}
\newtheorem{lemma}[lemma]{Lemma}
\aliascntresetthe{lemma}
\crefname{lemma}{Lemma}{Lemmas}

\newaliascnt{claim}{theorem}

\aliascntresetthe{claim}
\crefname{claim}{Claim}{Claims}

\newaliascnt{corollary}{theorem}
\newtheorem{corollary}[corollary]{Corollary}
\aliascntresetthe{corollary}
\crefname{corollary}{Corollary}{Corollaries}

\newaliascnt{proposition}{theorem}

\aliascntresetthe{proposition}
\crefname{proposition}{Proposition}{Propositions}

\newaliascnt{fact}{theorem}

\aliascntresetthe{fact}
\crefname{fact}{Fact}{Facts}

\newaliascnt{construction}{theorem}

\aliascntresetthe{construction}
\crefname{construction}{Construction}{Constructions}

\newaliascnt{conjecture}{theorem}

\aliascntresetthe{conjecture}
\crefname{conjecture}{Conjecture}{Conjectures}

\newaliascnt{definition}{theorem}
\newtheorem{definition}[definition]{Definition}
\aliascntresetthe{definition}
\crefname{definition}{Definition}{Definitions}

\newaliascnt{remark}{theorem}

\aliascntresetthe{remark}
\crefname{remark}{Remark}{Remarks}

\newaliascnt{observation}{theorem}
\newtheorem{observation}[observation]{Observation}
\aliascntresetthe{observation}
\crefname{observation}{Observation}{Observations}

\newcommand{\problem}[1]{\ensuremath{\textsc{#1}}\xspace}

\newcommand{\Collision}{\problem{Collision}}
\newcommand{\NestedCollision}{\problem{NestedCollision}}

\newcommand{\class}[1]{\ensuremath{\mathsf{#1}}\xspace}
\newcommand{\TFNP}{\class{TFNP}}

\newcommand{\PPAD}{\class{PPAD}}
\newcommand{\PPA}{\class{PPA}}
\newcommand{\PPP}{\class{PPP}}
\newcommand{\PWPP}{\class{PWPP}}
\newcommand{\NP}{\class{NP}}
\newcommand{\coNP}{\class{coNP}}

\newcommand{\PLS}{\class{PLS}}

\newcommand{\B}{\{0,1\}}
\newcommand{\N}{\mathbb{N}}

\newcommand{\poly}{\operatorname{poly}}

\newcommand{\Tainted}{\ensuremath{\mathsf{Tainted}(f,g)}\xspace}

\ifSUBMIT
\makeatletter

\let\CAS@orig@begintheorem\@begintheorem
\renewcommand{\@begintheorem}[2]{%
  \CAS@orig@begintheorem{#1}{#2}%
  \doublespacing
}

\let\CAS@origproof\proof
\let\CAS@origendproof\endproof
\renewenvironment{proof}{\CAS@origproof\doublespacing}{\CAS@origendproof}

\makeatother
\fi

\begin{document}
\let\WriteBookmarks\relax
\def\floatpagepagefraction{1}
\def\textpagefraction{.001}

\shorttitle{Black-Box PWPP Is Not Turing-Closed}

\shortauthors{P. Hubáček}

\title [mode = title]{Black-Box PWPP Is Not Turing-Closed}

\ifDRAFT
\tnotemark[1]

\tnotetext[1]{DRAFT -- PLEASE DO NOT DISTRIBUTE!}
\fi


\author[1,2]{Pavel Hubáček}
[ orcid=0000-0002-6850-6222
]

\cormark[1]


\ead{hubacek@math.cas.cz}



\affiliation[1]{organization={Institute of Mathematics, Czech Academy of Sciences},
addressline={Žitná 609/25},
city={Praha 1},
% %          citysep={}, % Uncomment if no comma needed between city and postcode
postcode={110 00},
country={Czech Republic}}

% \author[2]{}%[]

% Footnote of the second author

% Email id of the second author
% \ead{}



\affiliation[2]{organization={Charles University, Faculty of Mathematics and Physics},
addressline={Malostranské náměstí 25/2},
city={Praha 1},
%citysep={}, % Uncomment if no comma needed between city and postcode
postcode={118 00},
country={Czech Republic}}

\cortext[1]{Corresponding author}



\begin{abstract}
We establish that adaptive collision-finding queries are strictly more powerful than non-adaptive ones by proving that the complexity class \PWPP (Polynomial Weak Pigeonhole Principle) is not closed under adaptive Turing reductions in the black-box setting.
Previously, \PWPP was known to be closed under non-adaptive Turing reductions (Jeřábek 2016).
We demonstrate this black-box separation by introducing the \NestedCollision problem, a natural collision-finding problem defined on a pair of shrinking functions.
We show that while this problem is solvable via two adaptive calls to a \PWPP oracle, it cannot be solved via an efficient black-box non-adaptive reduction to the canonical \PWPP-complete problem \Collision.
\end{abstract}


\ifSUBMIT
\begin{highlights}
\item We show that the \TFNP subclass \PWPP is not closed under adaptive Turing reductions.
\item We introduce a natural nested collision-finding problem of independent interest.
\item Our separation is by a straightforward counting argument.
\end{highlights}
\fi

% Keywords
\begin{keywords}
 TFNP \sep PWPP \sep Turing reduction \sep adaptive \sep non-adaptive \sep closure \sep separation
\end{keywords}

\maketitle
\ifSUBMIT
\doublespacing
\AtBeginEnvironment{proof}{\doublespacing}
\fi

% Main text
\section{Introduction}
\label{sec:Introduction}

The complexity class \TFNP \citep{MegiddoP91} captures \NP search problems with a guaranteed existence of solutions.
Since \TFNP is a semantic class and is widely believed to lack complete problems (see, e.g., \cite{Pudlak2015}), its structural landscape is mapped through syntactic subclasses such as \PLS \citep{JohnsonPY88}, \PPA, \PPAD, \PPP, and \PWPP \citep{Papadimitriou94}.
A fundamental metric for evaluating the robustness of these subclasses is their closure properties under various notions of reducibility.
In particular, the closure under Turing reductions asks whether algorithms granted multiple adaptive oracle queries to a complete problem are strictly more powerful than those restricted to a single query.

Motivated by structural questions regarding reducibility first raised by~\cite{Papadimitriou94}, the foundational work by~\cite{BussJ12} revealed a dichotomy within \TFNP.
Topological and local-search classes like \PLS, \PPA, and \PPAD natively support sequential logic and are closed under fully adaptive Turing reductions.
Conversely, classes founded upon counting arguments and pigeonhole principles are less robust in terms of closure properties.
Resolving a long-standing conjecture in this line of work,~\cite{FlemingGPR24} recently established a black-box separation proving that \PPP is not closed under Turing reductions.
The capability of a class to simulate sequential dependency is highly relevant to the algorithmic aspects of extremal combinatorics, where natural search problems like the one motivated by Ramsey's theorem inherently require iteratively applied, adaptive pigeonhole arguments.
Recent separations by~\cite{JainLRX24} regarding generalized iterated pigeonhole classes highlight the power of sequential applications of the pigeonhole principle in the context of \TFNP.

\paragraph{The curious case of \PWPP.}
The complexity class \PWPP (Polynomial Weak Pigeonhole Principle,~\citealp{Jerabek16}) occupies an intermediate regime.
Unlike \PPP, its analogue for length-preserving functions motivated by the strong pigeonhole principle, \PWPP is defined via the weak pigeonhole principle in the strictly shrinking regime.
This was leveraged by~\cite{Jerabek16} to establish that multiple independent \PWPP instances can be solved via a single call to a \PWPP oracle, proving that \PWPP is closed under non-adaptive Turing reductions.
Motivated by the work of~\cite{Korten22},~\cite{LiPT24} recently defined a subclass of \PWPP called \class{LOSSY} and showed that it is Turing-closed.
However, it has remained an open question whether \PWPP is closed under adaptive Turing reductions.
Importantly, the techniques from~\cite{FlemingGPR24} for the Turing closure of \PPP are not helpful in studying the adaptive Turing closure of \PWPP as they are tailored to establishing the strictly increasing power of additional \emph{non-adaptive} queries, whereas, in \PWPP, additional non-adaptive queries are not helpful as established by~\cite{Jerabek16}.

\paragraph{Our result.}
We show that black-box non-adaptive access to a \PWPP oracle is insufficient to simulate even two adaptive collision-finding steps.
Concretely, we define a natural search problem \NestedCollision on a pair of shrinking functions and prove that it lies in the adaptive Turing closure of \PWPP but does not admit an efficient black-box many-one reduction to the \PWPP-complete problem \Collision.
Throughout, ``black-box'' means we operate in the decision-tree model for search problems.
In this setting, the reduction accesses the input instances exclusively via shallow decision trees that query individual truth-table entries.
It consists of a forward part that computes the entries of an induced \Collision instance, and a backward part that must produce a \NestedCollision solution using only its own queries together with a provided \Collision witness. Crucially, the backward part must succeed given an arbitrary valid witness to the induced \Collision instance.
At a high level, the separation follows a simple principle. A witness for a \Collision instance given by a shallow decision-tree formulation induces a small transcript of oracle queries.
By analyzing random instances of \NestedCollision, we show that, for all large enough input lengths, there always exists a valid \Collision witness whose joint transcript misses at least one ingredient needed to assemble a full \NestedCollision solution.
Because the transcript is incomplete and the backward reduction is shallow, the reduction is forced to output a candidate solution whose correctness is not yet logically forced.
Via an adversarial argument, we then complete the unqueried entries of the functions, rendering this candidate output invalid and defeating the reduction.

Our separation between \Collision and \NestedCollision pins down the power of adaptive collision-finding queries by establishing a black-box separation between \PWPP and its adaptive Turing closure.
Our results isolate \PWPP in a unique intermediate tier among the standard \TFNP subclasses with respect to closure properties.
Namely, \PWPP is strictly more resilient than \PPP due to robustness with respect to non-adaptive Turing reductions, yet structurally different from the major part of \TFNP that is closed under adaptive Turing reductions.
Extending our result to separate the non-adaptive and adaptive Turing closures of \PPP is left for future work.

\subsection{Related Work}

\paragraph{Complete problems for \PPP and \PWPP.}
For many years, a major open question was whether the pigeonhole classes \PPP and \PWPP contained other complete problems in addition to the canonical ones defining the two classes.
\cite{SotirakiZZ18} provided the first additional complete problems for these classes, demonstrating connections to lattice-based cryptography through the constrained-SIS problem.
Exploring the internal structure of these classes further,~\cite{BanJPPR19} established several reductions between problems related to integral lattices, showing that, e.g., the computational analogue of Minkowski's theorem lies in \PPP.
Motivated by the connections between \PWPP and collision-resistant hashing,~\cite{HubacekV21} expanded the algebraic repertoire of these classes by analyzing the search complexity of the discrete logarithm problem, establishing that suitable variants of it over general groupoids exactly capture the complexity of \PPP and \PWPP.
Parallel to cryptographic reductions, a recent line of work has closely linked \PWPP and \PPP to existential guarantees in extremal combinatorics.
Building on earlier results showing the \PWPP-hardness of implicit versions of Ramsey's theorem and the Erd\H{o}s-Rado sunflower lemma \citep{Krajicek05,KomargodskiNY19},~\cite{BourneufFHRS23} showed that several other classical theorems from extremal combinatorics, such as the Erd\H{o}s-Ko-Rado theorem, Sperner's theorem, or Cayley's formula, give rise to natural complete problems for \PWPP and \PPP.
Similar characterizations were independently established by~\cite{PasarkarPY23}.
\cite{Ishizuka25} extended this line of work by showing that a search problem corresponding to K{\H{o}}nig's lemma is contained in $\PPA\cap\PPP$.
Along a different axis of generalization,~\cite{BennettGS25} studied the complexity of finding $L$-wise \emph{multicollisions}, organizing the resulting multi-pigeonhole classes and connecting them to coding and lattice problems; in contrast to our setting, their generalization varies the number of colliding elements rather than introducing adaptive, sequential collision-finding.

\paragraph{Black-box separations and proof complexity.}
Because separating \TFNP subclasses in the standard Turing machine model would imply $\mathsf{P}\neq\NP\cap\coNP$, unconditional separations are typically studied in the black-box (or relativized) setting.
The formal study of black-box reductions in \TFNP was pioneered by~\cite{BeameCEIP98}, who established a foundational framework connecting oracle separations between total search problems to lower bounds in propositional proof complexity.
In recent years, this proof-complexity lens has driven a line of work refining the \TFNP landscape.
Notably,~\cite{GoosHJMPRT24} leveraged lower bounds in systems like Nullstellensatz and Sherali-Adams to complete the picture of black-box relationships among all classical \TFNP subclasses, but there are many other works establishing separations within subclasses of \TFNP building on this rich interplay (see, e.g., \cite{BureshOppenheimM04,GargGK020,BussFI23,HubacekKT24,LiPR24}).
We note that our separation between \PWPP and its adaptive Turing closure is conceptually much simpler compared to the refinements of the \TFNP landscape based on proof complexity.

Finally,~\cite{Jerabek07} established relativized separations in bounded arithmetic comparing various extended notions of the weak pigeonhole principle. While he established that adaptive applications of the \emph{surjective} weak pigeonhole principle escape the provability of weak base theories, he notably proved that certain adaptive iterations of the \emph{injective} weak pigeonhole principle (i.e., the proof-complexity analogue of \PWPP) are provable.
Importantly, his proof-theoretic separations are unrelated to the separation between the adaptive and non-adaptive Turing closure of \PWPP, and, in a sense, serve as an additional motivation for our result.

\section{Preliminaries}
\label{sec:Preliminaries}
\begin{definition}[\TFNP and many-one reductions]
A poly-time relation $R\subseteq\B^*\times\B^*$ is in \class{FNP} if there is a polynomial $p$ such that $\forall x, y\in\B^*, (x,y) \in R \Rightarrow |y| \leq p(|x|).$
A relation $R \in \class{FNP}$ is in \TFNP if it satisfies
$\forall x\in\B^* \exists y\in\B^*, (x,y) \in R$.

We say that $A \in \TFNP$ \emph{many-one reduces} to $B \in \TFNP$ (written also as $A \leq B$) if there exist two poly-time computable functions $f$ and $g$ satisfying
$\forall x, y \in \{0, 1\}^*, (f(x), y) \in B \Rightarrow (x, g(x,y)) \in A.$
In contrast, $A$ is \emph{Turing reducible} to $B$ (written also as $A \leq_T B$) if there exists a polynomial-time oracle Turing machine that solves $A$ by making multiple, sequential queries to an oracle for $B$, where the input to each query can depend on the answers received from previous ones.
\end{definition}

\begin{definition}[\Collision and \PWPP]
\label{def:Collision-PWPP}
The search problem \Collision is defined via the following relation of instances and solutions.
\begin{description}
\item[Instance:] A Boolean circuit $C\colon\B^n\to\B^{n-1}$.
\item[Solution:] A pair of inputs $x_1, x_2 \in \{0,1\}^n$ such that $x_1 \neq x_2$ and $C(x_1) = C(x_2)$.
\end{description}
The class of all total search problems many-one reducible to \Collision is called \PWPP.
\end{definition}

\section{Nested Collision-Finding}
\label{sec:Nested-Collision}

Throughout this section and the lower-bound section, we identify \(\B^{kn}\) with \((\B^n)^k\) in the natural blockwise way whenever named \(n\)-bit blocks are used.
Thus, for example, \(z=(x,y)\in\B^{2n}\) means \(x,y\in\B^n\), and \((x,y,u,v)\in\B^{4n}\) denotes the corresponding four-block string.

\begin{definition}[\NestedCollision]
\label{def:Nested-Collision}
The search problem \NestedCollision is defined via the following relation of instances and solutions.
\begin{description}
\item[Instance:] A pair of Boolean circuits
$f: \B^n \to \B^{n-1}$
and
$g: \B^{2n} \times \B^n \to \B^{n-1}$.
\item[Solution:] A tuple $(x, y, u, v)\in\B^{4n}$ such that the following hold:
\begin{enumerate}
\item $x \neq y$ and $f(x)=f(y)$, and
\item $u \neq v$ and $g(x,y,u) = g(x,y,v)$.
\end{enumerate}
\end{description}
\end{definition}

To simplify notation, for $z=(x,y)\in\B^{2n}$, we denote by $g_{x,y}(\cdot)=g(x,y,\cdot)=g_z(\cdot)=g(z,\cdot)\colon\B^{n}\to\B^{n-1}$ the restriction of $g$ obtained by fixing the first two coordinates to $z$.
The \NestedCollision problem asks to find a collision $z=(x,y)$ in $f$ and a collision in the restriction $g_z$.
Since both functions are shrinking, such a pair of collisions always exists and \NestedCollision is total.
Notice that \NestedCollision is clearly in the adaptive Turing closure of \PWPP.
Given an instance $(f,g)$ of \NestedCollision, we can solve it by two adaptive calls to a $\PWPP$ oracle by first finding a non-trivial collision $(x,y)$ in $f$ and then a non-trivial collision $(u,v)$ in $g_{x,y}$.

\begin{observation}
$\NestedCollision\le_T\Collision$.
\end{observation}

The discussion so far uses the standard circuit, or white-box, model: the functions $f$ and $g$ are given by Boolean circuits, and the above observation is an ordinary adaptive Turing reduction to the \PWPP-complete problem \Collision.
The lower bound presented below is in the black-box model.
In that setting, the same functions are represented explicitly, but the reduction, modeled by decision trees, accesses them only through evaluation queries.

\section{Nested Collision-Finding Is Strictly Harder in the Decision-Tree Model}

In this section, we show that \NestedCollision does not admit a shallow many-one decision-tree reduction to \Collision.
Here, we do not restrict the representations of the input functions by circuits.
Instead, we work with the decision-tree model formulations of the search problems, where an instance of \(\NestedCollision_n\) is identified with the full ``truth tables'' of functions
\[
f\colon \B^n\to\B^{n-1}
\qquad\text{and}\qquad
g\colon \B^{2n}\times\B^n\to\B^{n-1},
\]
and the reduction is allowed to access these functions only through oracle queries to individual table entries.
Thus, the forward and backward parts of the reduction are modeled as decision trees whose internal nodes query values of \(f(x)\) or \(g(x,y,u)\), and whose leaves output, respectively, a value of the induced \(\Collision_m\) instance and a candidate \(\NestedCollision_n\) solution.
The following is a specialization of the standard notion of many-one decision-tree reductions between total search problems to the case of reductions from \NestedCollision to \Collision.

\begin{definition}[\Collision-formulation of \NestedCollision]
\label{def:collision-formulation-of-nested-collision}
Let $n,m,d\in\N$.
A depth-$d$ $\Collision_m$-formulation of $\NestedCollision_n$ consists of two families of decision trees $\left(T_w\right)_{w\in\B^m}$ and $\left(D_{w_1,w_2}\right)_{w_1\neq w_2\in\B^m}$ such that:
\begin{description}
\item[Forward part:] For every $w\in\B^m$, $T_w$  is a decision tree of depth at most $d$, and, given an input instance $(f,g)$ of $\NestedCollision_n$, outputs a value
\(
C^{f,g}(w)\in\B^{m-1}.
\)
\item[Backward part:] For every pair $w_1\neq w_2\in\B^m$, $D_{w_1,w_2}$ is a decision tree  of depth at most $d$, and, given an input instance $(f,g)$ of $\NestedCollision_n$, outputs a tuple
\(
(x,y,u,v)\in\B^{4n}.
\)
\item[Correctness:]
For every instance $(f,g)$ of $\NestedCollision_n$ and every pair $w_1\neq w_2\in\B^m$, if
\(
C^{f,g}(w_1)=C^{f,g}(w_2),
\)
then
\(
D_{w_1,w_2}(f,g)=(x,y,u,v)
\)
is a solution to the $\NestedCollision_n$ instance $(f,g)$.
\end{description}
\end{definition}

\paragraph{Endpoint-revealing forward trees.}
We say that the forward part of a \Collision-formulation is \emph{endpoint-revealing} if, on every root-to-leaf path of every tree $T_w$, whenever the path queries an entry $g(x,y,u)$, the same path also queries the two endpoint entries $f(x)$ and $f(y)$.
The order of these endpoint queries is immaterial for the argument below.

\begin{lemma}[endpoint-revealing normal form]
\label{lem:endpoint-revealing-normal-form}
Every depth-$d$ $\Collision_m$-formulation of $\NestedCollision_n$ can be transformed into a depth-$3d$ $\Collision_m$-formulation of $\NestedCollision_n$ whose forward part is endpoint-revealing and that computes the same induced \Collision instance.
Moreover, this transformation leaves all backward trees unchanged.
\end{lemma}

\begin{proof}
For every node of a forward tree labeled by a query $g(x,y,u)$, replace the node by three consecutive queries $f(x)$, $f(y)$, and $g(x,y,u)$.
All outgoing edges of the two inserted $f$-query nodes lead to the same next inserted query, so their answers are ignored.
After the final $g$-query, the transformed tree continues exactly as the original tree did.
This increases the depth of each forward tree by at most a factor of $3$, while all backward trees are left unchanged.
The value output by each forward tree on every instance is unchanged.
Thus the induced map $C^{f,g}$ is unchanged, and correctness is preserved.
\end{proof}

\begin{theorem}
\label{thm:NestedCollision-vs-Collision}
Let $d,m\colon\N\to\N$ be polynomially bounded functions.
Then, for all sufficiently large $n$, $\NestedCollision_n$ admits no depth-$d(n)$ $\Collision_{m(n)}$-formulation.
\end{theorem}

\paragraph{Roadmap of the lower bound.}
Suppose we are given a depth-$d$ $\Collision_m$-formulation of $\NestedCollision_n$.
For an instance $(f,g)$ of $\NestedCollision_n$, the forward part $\left(T_w\right)_{w\in\B^m}$ computes a shrinking function
\(
C^{f,g}\colon \B^m \to \B^{m-1}.
\)
For an input $w\in\B^m$, evaluating $C^{f,g}(w)$ reveals a small transcript of queries to $f$ and $g$.
By \Cref{lem:endpoint-revealing-normal-form}, we may assume that the forward part is endpoint-revealing at the cost of only a constant-factor increase in depth.
Thus, whenever a transcript touches a slice $g_{x,y}$, it also reveals the two endpoint values $f(x)$ and $f(y)$.
We call a transcript \emph{tainted} if it already contains a collision in $f$.
We show that for a uniformly random instance $(f,g)$, most inputs $w$ are \emph{good} (i.e., not tainted) with overwhelming probability.
Hence, by the pigeonhole principle, there exists a \Collision witness $\{w_1,w_2\}$ consisting of two good inputs.
The endpoint-revealing property ensures that the joint transcript of two good inputs is non-witnessing.
Therefore, the backward tree $D_{w_1,w_2}$ associated with the collision witness would have to solve \NestedCollision on all completions of such a non-witnessing partial assignment.
We then show by an adversarial completion argument that no decision tree of depth $d$ can do this.
This yields the contradiction.

\paragraph{Query sets of forward evaluation trees.}
Let $\left(T_w\right)_{w\in\B^m}$ and $\left(D_{w_1,w_2}\right)_{w_1\neq w_2\in\B^m}$ be a candidate depth-$d$ $\Collision_m$-formulation of $\NestedCollision_n$.
For a $\NestedCollision_n$ instance $(f,g)$, let $Q(T_w,f,g)$ denote the partial assignment to entries of $f$ and $g$ revealed along the execution of $T_w$ on $(f,g)$.
Additionally, denote by $Q_f(T_w,f,g)$ the set of all $f$-queries in the domain of $Q(T_w,f,g)$.
The domain of such a partial assignment consists of $f$-locations $x\in\B^n$ and $g$-locations $(x,y,u)\in\B^{2n}\times\B^n$.
We say that a transcript contains a query $x$ or $(x,y,u)$ when the corresponding location lies in its domain.
Unions of transcripts are understood as consistent unions of partial assignments, and $|P|$ denotes the size of the domain of a partial assignment $P$.

Next, we define tainted inputs with respect to a $\NestedCollision$ instance $(f,g)$.

\begin{definition}[tainted inputs]
\label{def:tainted-inputs}
Let $T_w$ be the forward evaluation tree associated with $w\in\B^m$.
We say that $w$ is a \emph{tainted input for $f$ and $g$} if there exist distinct queries $x,y\in Q_f(T_w,f,g)$ such that $f(x)=f(y)$.
We denote by $\Tainted\subseteq\B^m$ the set of all tainted inputs for $f$ and $g$.
\end{definition}

Correspondingly, we call an input \emph{good} if it is not tainted.
Note that, although the defining collision condition refers only to values of $f$, the set $Q_f(T_w,f,g)$ consists of the $f$-queries made along the actual adaptive execution of $T_w$ on the oracle pair $(f,g)$, and, hence, may depend on both $f$ and $g$.
On the other hand, since we fix the $\Collision$ formulation throughout the analysis, while only the functions $(f,g)$ vary, we keep the dependence on $(f,g)$ explicit and suppress the dependence on the forward trees from the notation.

We use the following terminology for partial assignments.
A partial assignment $P$ to entries of the truth tables of $f$ and $g$ \emph{witnesses a full $\NestedCollision_n$ solution} if there exist $x,y,u,v\in\B^n$ with $x\neq y$ and $u\neq v$ such that $P$ assigns values to $f(x)$, $f(y)$, $g(x,y,u)$, and $g(x,y,v)$, and these assigned values satisfy $f(x)=f(y)$ and $g(x,y,u)=g(x,y,v)$.
If no such tuple exists, we call $P$ \emph{non-witnessing}.

Next, we argue that two good inputs do not help in solving the \NestedCollision instance $f,g$.

\begin{lemma}[two good inputs give a non-witnessing transcript]
\label{lem:good-pair-no-solution}
Let $(f,g)$ be any instance of $\NestedCollision_n$.
Assume that the forward part of its $\Collision_m$-formulation is endpoint-revealing.
Let $w_1,w_2\in\B^m$ be a collision for $C^{f,g}$ such that both inputs are good.
Then the joint transcript
\(
Q(T_{w_1},f,g)\cup Q(T_{w_2},f,g)
\)
is non-witnessing.
\end{lemma}

\begin{proof}
Suppose that the joint transcript witnesses a full $\NestedCollision_n$ solution.
Then there exist $x,y,u,v\in\B^n$ such that
\[
x\neq y,
\qquad
f(x)=f(y),
\qquad
u\neq v,
\qquad
g_{x,y}(u)=g_{x,y}(v).
\]
Moreover, the joint transcript assigns values to $f(x)$, $f(y)$, $g(x,y,u)$, and $g(x,y,v)$.
At least one of the two transcripts contains one of the queries $(x,y,u)$ and $(x,y,v)$.
By the endpoint-revealing property, that transcript also contains the two queries $x$ and $y$ to $f$.
Since $x\neq y$ and $f(x)=f(y)$, this transcript is tainted.
This contradicts the assumption that both inputs are good.
\end{proof}

Now, we proceed to argue that there exists a pair of functions $(f,g)$ that induces a collision between two good inputs.
Our argument uses a probabilistic argument over uniform $\NestedCollision_n$ instances defined below.

\begin{definition}[uniform nested functions]
\label{def:nested-collision-hard-distribution}
Let $\mathcal{O}$ be the distribution of pairs of functions
\[
f\colon \B^n \to \B^{n-1}
\qquad\text{and}\qquad
g\colon \B^{2n}\times\B^n \to \B^{n-1}
\]
obtained by sampling both $f$ and $g$ uniformly and independently.
\end{definition}

Equivalently, under $\mathcal{O}$, for every $z\in\B^{2n}$, the slice $g_z\colon\B^n\to\B^{n-1}$ is a uniformly random function, and these slices are mutually independent and independent of $f$.

Next, we bound the expected number of tainted inputs when $(f,g)$ are sampled from $\mathcal{O}$, denoted $(f,g)\gets\mathcal{O}$.

\begin{lemma}[bounding the number of tainted inputs]
\label{lem:bound-tainted}
Let $n,m,d\in\N$ and fix a depth-$d$ $\Collision_m$-formulation of $\NestedCollision_n$.
Then
\(
\mathbb{E}_{f,g\gets\mathcal{O}}\bigl[|\Tainted|\bigr]
\le
d^2\cdot 2^{m-n}.
\)
\end{lemma}

\begin{proof}
Fix an input $w\in\B^m$.
The tree $T_w$ makes at most $d$ queries.
Hence, the transcript $Q(T_w,f,g)$ contains at most $d$ revealed values and, in particular, at most $d$ queries to $f$.
Even though the queries are adaptive, every fresh query to $f$ returns an independent uniform value in $\B^{n-1}$ conditioned on the previous transcript.
Since $f$ is a uniform shrinking function, the probability that the transcript contains a collision in $f$ is, by the union bound, at most
\[
\binom{d}{2}2^{-(n-1)}
\le
d^2 2^{-n}.
\]
Summing over all $2^m$ inputs gives
\[
\mathbb{E}_{f,g\gets\mathcal{O}}\bigl[|\Tainted|\bigr]
\le
2^m\cdot d^2 2^{-n}
=
d^2\cdot 2^{m-n}.\qedhere
\]
\end{proof}

Equipped with the above bound, we can prove the existence of a good collision for all large enough $n$.

\begin{lemma}[existence of a good collision]
\label{lem:almost-always-good-collision-exists}
Let $n,m,d\in\N$ and fix a depth-$d$ $\Collision_m$-formulation of $\NestedCollision_n$.
Let $(f,g)\gets\mathcal{O}$.
Except with probability at most
\(
d^2 2^{-(n-1)},
\)
there exists a non-trivial collision $w_1,w_2\in\B^m$ for $C^{f,g}$ such that both $w_1$ and $w_2$ are good.
\end{lemma}

\begin{proof}
By \Cref{lem:bound-tainted}, we have
\[
\mathbb{E}_{f,g\gets\mathcal{O}}[|\Tainted|]
\le
d^2 2^{m-n}.
\]
By Markov's inequality,
\[
\Pr_{f,g\gets\mathcal{O}}\left[|\Tainted|\ge 2^{m-1}\right]
\le
d^2 2^{m-n}\cdot 2^{-(m-1)}
=
d^2 2^{-(n-1)}.
\]
Therefore, with probability at least
\(
1-d^2 2^{-(n-1)},
\)
we have $|\Tainted|<2^{m-1}$.
Since $|\Tainted|$ is an integer, this implies $|\Tainted|\le 2^{m-1}-1$.
Thus, the set of good inputs has size at least $2^m-(2^{m-1}-1)=2^{m-1}+1$.
Since
\(
C^{f,g}\colon\B^m\to\B^{m-1}
\)
is shrinking, the pigeonhole principle implies that two good inputs collide.
\end{proof}

\begin{lemma}[lower bound for any non-witnessing partial assignment]
\label{lem:restricted-lb}
Let $P$ be a partial assignment to entries of the truth tables of $f$ and $g$.
Assume that $P$ is non-witnessing.
Let $A$ be a decision tree of depth at most $d$ querying entries of the truth tables of $f$ and $g$.
If $n\ge 2$ and $|P|+d \le 2^{n-1}$, then there exists a completion of $P$ on which the output of $A$ is not a valid $\NestedCollision_n$ solution.
\end{lemma}

\begin{proof}
We answer the queries of $A$ adversarially while extending the partial assignment $P$.
Queries to already fixed entries are answered consistently with the current partial assignment.

Whenever $A$ queries a fresh $f$-input $a\in\B^n$, we answer with a value in $\B^{n-1}$ that has not previously been used on any queried $f$-input.
Such a value exists: apart from $a$ itself, at most $|P|+d-1$ entries of $f$ are fixed at this point, so at most $|P|+d-1\le 2^{n-1}-1$ values of $\B^{n-1}$ are excluded and at least one remains free.

Whenever $A$ queries a fresh $g$-input $(z,u)\in\B^{2n}\times\B^n$, we answer with a value in $\B^{n-1}$ that has not previously been used on any queried point of the same slice $g_z$.
Again such a value exists, since apart from $(z,u)$ at most $|P|+d-1\le 2^{n-1}-1$ points of the slice $g_z$ are fixed, leaving at least one free value of $\B^{n-1}$.
This condition on the slices is still needed because $P$ is only assumed to be non-witnessing.
It may already contain an $f$-collision, and creating a new collision in the corresponding slice $g_z$ could then create a full $\NestedCollision_n$ witness.

Thus, no answer given during the adversarial simulation creates an $f$-collision involving a freshly answered $f$-query, and no answer creates a collision in any slice $g_z$ involving a freshly answered $g$-query.
Existing collisions already present in $P$ are allowed.

In particular, if the final transcript is denoted by $P'$, then $P'$ is still non-witnessing.
Indeed, any full witness in $P'$ was not already present in $P$, since $P$ is non-witnessing.
Hence, at least one of the two required collisions would have to be newly created during the simulation, contradicting the way fresh $f$- and $g$-queries were answered.

Let $A(P')=(x,y,u,v)$ be the tuple output at the reached leaf, where $P'\supseteq P$ denotes the final transcript.
If it is invalid for a trivial reason, e.g.\ $x=y$ or $u=v$, we are done; so assume $x\neq y$ and $u\neq v$.
Since $P'$ is non-witnessing, at least one of the two equalities $f(x)=f(y)$ and $g(x,y,u)=g(x,y,v)$ is not forced by $P'$.
Fix one such unforced equality and call the two entries it compares the \emph{relevant} entries.
If both relevant entries are already revealed in $P'$, then, since the equality is not forced, they are unequal and the output is already invalid.
Otherwise at least one relevant entry is still undefined.
If exactly one relevant entry is undefined, we assign it any value distinct from the value of the already defined partner entry.
If both relevant entries are undefined, we assign them two distinct values.
In either case this is possible because $n\ge2$, and hence $\B^{n-1}$ contains at least two values; in contrast to the run, these final values need not avoid collisions, as no further queries are answered.
The chosen completion makes the relevant equality fail, so $(x,y,u,v)$ is not a valid $\NestedCollision_n$ solution.
All remaining undefined entries are completed arbitrarily, yielding a full completion of $P$ that falsifies the output of $A$.
\end{proof}

\begin{proof}[Proof of \Cref{thm:NestedCollision-vs-Collision}]
Fix polynomially bounded functions $d,m\colon\N\to\N$.
We show that, for all sufficiently large $n$, no such formulation exists.
Fix such an $n$ and suppose to the contrary that there exists a depth-$d(n)$ $\Collision_{m(n)}$-formulation of $\NestedCollision_n$.
For readability, write $d=d(n)$ and $m=m(n)$.
Let $\Delta=3d$.
By \Cref{lem:endpoint-revealing-normal-form}, we obtain an endpoint-revealing depth-$\Delta$ $\Collision_m$-formulation of $\NestedCollision_n$ computing the same induced \Collision instance.
We work with this normalized formulation from now on.
Let $C^{f,g}\colon\B^m\to\B^{m-1}$ be the shrinking map computed by the formulation on $(f,g)$.
For each $w\in\B^m$, let $T_w$ be the forward evaluation tree computing $C^{f,g}(w)$.

By \Cref{lem:almost-always-good-collision-exists}, except with probability at most $\Delta^2 2^{-(n-1)}$ over $(f,g)\gets\mathcal{O}$, there exists a collision $w_1,w_2\in\B^m$ for $C^{f,g}$ such that both inputs are good.
For all sufficiently large $n$, this probability is strictly smaller than $1$.
Fix an instance $(f,g)$ for which such a good collision $w_1,w_2$ exists.
Let
\[
P:=Q(T_{w_1},f,g)\cup Q(T_{w_2},f,g)
\]
be the joint transcript of the two forward evaluation trees on this instance.
The partial assignment $P$ has size at most $2\Delta$.
Moreover, by \Cref{lem:good-pair-no-solution}, $P$ is non-witnessing.

Note that $w_1\neq w_2$ and $C^{f',g'}(w_1)=C^{f',g'}(w_2)$ hold for every completion $(f',g')$ extending $P$.
Indeed, the forward trees $T_{w_1}$ and $T_{w_2}$ follow the same root-to-leaf paths on every such completion, because all queried entries along those paths are fixed by $P$.
Therefore, by correctness of the formulation, the backward part tree $D_{w_1,w_2}$ must output a valid $\NestedCollision_n$ solution on every completion of $P$.
The backward tree was not changed by the normalization, and, hence, has depth at most $d\le\Delta$.
Since $\Delta\in\poly(n)$, for all sufficiently large $n$ we have $n\ge 2$ and $|P|+\Delta\le 3\Delta\le 2^{n-1}$.
Applying \Cref{lem:restricted-lb} to the partial assignment $P$ and the decision tree $D_{w_1,w_2}$ with depth parameter $\Delta$, we conclude that there exists a completion $(f',g')$ of $P$ on which $D_{w_1,w_2}$ outputs an invalid solution for the $\NestedCollision_n$ instance $(f',g')$.
This contradicts the correctness of the formulation.
Therefore, no such depth-$d$ \Collision-formulation exists.
\end{proof}

As observed in \Cref{sec:Nested-Collision}, the \NestedCollision problem is contained in the adaptive Turing closure of \PWPP.
On the other hand, viewing black-box many-one reductions in the standard way as polynomial-depth decision-tree formulations, \Cref{thm:NestedCollision-vs-Collision} shows that \NestedCollision is not black-box many-one reducible to \Collision.
Since \PWPP is closed under \emph{non-adaptive} Turing reductions \citep{Jerabek16}, and this closure is established by a black-box construction combining many non-adaptive \Collision instances into a single one, it follows that \NestedCollision is not black-box many-one reducible to the \emph{non-adaptive} Turing closure of \PWPP.
This establishes our main result.

\begin{corollary}
\label{cor:pwpp-not-turing-closed}
\PWPP is not closed under adaptive Turing reductions in the black-box setting.
\end{corollary}

\section*{Declaration of generative AI and AI-assisted technologies in the manuscript preparation process}

During the preparation of this work, the author used Bolzano \citep{Bolzano26}, an LLM-based multi-agent system for mathematical research, to assist in exploring candidate black-box separations and in drafting an early version of the proof of the main theorem. This statement also records the extent of the tool-assisted research interaction.

The author asked Bolzano to search for a total search problem separating \PWPP from its adaptive Turing closure in the black-box setting. In this interaction, Bolzano proposed the \NestedCollision problem as a candidate. The lower-bound proof initially produced by Bolzano was incorrect. The author then supplied the high-level strategy for the lower bound: to show that, with high probability over a random instance, the induced \Collision instance has a valid witness whose query transcript does not reveal enough information to construct a \NestedCollision solution. Bolzano produced a draft proof following this strategy, using a counting argument based on Markov's inequality.

The final manuscript was written, checked, restructured, and rewritten by the author. In particular, after a colleague identified a gap in an earlier random-oracle framing, the author reformulated the result directly in the decision-tree model, the standard setting for black-box separations in \TFNP. Further simplifications were incorporated in response to reviewer feedback. The remaining contributions from the Bolzano-assisted interaction are the formulation of the \NestedCollision problem (\Cref{def:Nested-Collision}) and the probabilistic structure used to bound the number of tainted inputs (\Cref{lem:bound-tainted,lem:almost-always-good-collision-exists}).

The author independently reviewed and verified all AI-assisted material, checked the relevant sources and attributions, and thoroughly edited, rewrote, and adapted the material. The manuscript represents the author's authentic and original contribution and reflects the author's own analysis, interpretation, insights, and ideas. The author takes full responsibility for the correctness, originality, and content of the manuscript. Bolzano is not an author and bears no responsibility for the manuscript. A detailed description of the Bolzano system and an account of this interaction, including a link to the research transcript, can be found in \citep{BalkoGHKKRSZ26}.

\section*{Acknowledgments}

I am grateful to Pavel Pudlák and Neil Thapen for insightful discussions on the adaptive Turing closure of the complexity classes \PPP and \PWPP, and for their helpful comments on a draft of this manuscript.
I am also grateful to Vašek Rozhoň for providing access to the LLM-based service Bolzano and letting me experiment with it.
I sincerely thank Weiqiang Yuan for identifying a subtle gap related to the black-box reduction formalism in a prior version of this manuscript.
Finally, I wish to thank the anonymous reviewers for their constructive feedback on the presentation of the result.

\emph{Funding:} This work was partially supported by the Academy of Sciences of the Czech Republic (RVO 67985840) and the Czech Science Foundation GAČR grant No. 25-16311S.

\printcredits

\bibliographystyle{cas-model2-names}
\bibliography{bibliography}
\end{document}